\documentclass[11pt]{article}

\usepackage{fullpage}
\usepackage{latexsym}
\usepackage{amssymb}
\usepackage{amsfonts}

\def\01{\{0,1\}}

\newcommand{\Oh}[1]{\mathrm{O}\!\left(#1\right)}

\newcommand{\oo}[1]{\mathrm{o}(#1)}
\newcommand{\Om}[1]{\Omega\!\left(#1\right)}
\newcommand{\Th}[1]{\Theta\!\left(#1\right)}

\newcommand{\OR}{\mbox{\rm OR}}
\newcommand{\Search}{\mbox{\rm SEARCH}}

\newcommand{\DISJ}{\mbox{\rm DISJ}}
\newcommand{\NDISJ}{\mbox{\rm NDISJ}}

\newcommand{\vectorres}[2]{{#1^{(#2)}}}

\newtheorem{definition}{Definition}
\newtheorem{theorem}{Theorem}
\newtheorem{fact}[theorem]{Fact}
\newtheorem{lemma}[theorem]{Lemma}

\newenvironment{proof}[1][Proof.]{
    \par
    \noindent \textbf{#1}
}{
    \nobreak\leavevmode
    \hfill $\Box$\par\bigskip
}

\begin{document}

\title{\bf A Strong Direct Product Theorem for Disjointness}
\author{Hartmut Klauck\\
Centre for Quantum Technologies and\\
     SPMS, Nanyang Technological University\\
       Singapore\\
hklauck@gmail.com
}
\date{}
\maketitle

\begin{abstract}
A strong direct product theorem states that if we want to compute
$k$ independent instances of a function, using less than $k$ times
the resources needed for one instance, then the overall success
probability will be exponentially small in $k$.
We establish such a theorem for the randomized communication complexity
of the Disjointness problem, i.e., with communication $const\cdot kn$ the success probability
of solving $k$ instances of size $n$ can only be exponentially small in $k$. This solves an open problem of
\cite{klauck:qsdpt,lee:dptdisc}. We also show that this bound even holds for $AM$-communication protocols with limited ambiguity.

The main result implies a new lower bound for Disjointness in a restricted 3-player NOF protocol,
and optimal communication-space tradeoffs for Boolean matrix product.

Our main result follows from a solution to the dual of a linear programming problem,
whose feasibility comes from a so-called Intersection Sampling Lemma that generalizes
a result by Razborov \cite{razborov:disj}.
\end{abstract}

\section{Introduction}

\subsection{Direct product theorems}

One of the fundamental questions that can be asked in any model of computation is how well computing several instances of the
same problem can be composed. Are significant savings possible when computing the same function $f$ on $k$ independent inputs?
Or is it true that the optimal way to do this is to run the same algorithm independently $k$ times?

This question can be asked for any measure of complexity and is usually referred to as the {\em direct sum problem}.
In this paper we consider the model of randomized communication complexity.
A protocol between players Alice and Bob is given $k$ inputs $(x_1,y_1),\ldots,(x_k,y_k)$, and has
to output the vector of $k$ answers $f(x_1,y_1),\ldots,f(x_k,y_k)$.
The question is how the protocol can optimally distribute its
resources among the $k$ input instances it needs to compute.
In this setting we are not only concerned with the overall communication, but also with the achievable success probability.
If the trivial way of computing $f$ for $k$ input instances by running the same protocol $k$ times independently is really optimal,
then we should expect the success probability $\sigma$ to go down exponentially with $k$.

Such statements can come in two flavors. First, if every protocol with communication $c$ has constant error probability when
computing just \emph{one} instance of $f$, then for computing $k$ instances with communication $c$ we expect $\sigma$ to be exponentially small in $k$. If this is the case for a function $f$, we say that a \emph{weak} direct product theorem holds for $f$. If this is the case for {\em all} functions, we say that a weak direct product theorem holds in general.

However, even if we allow our protocol to use communication $kc$ we might expect $\sigma$ to be exponentially small in $k$,
unless the protocol could somehow correlate its computation on several instances for all possible choices of inputs.
If such a statement is true we call it a \emph{strong} direct product theorem (SDPT).

Strong direct product theorems are usually hard to prove and sometimes not even true.
In particular Shaltiel \cite{shaltiel:sdpt} exhibits a general setup
in which strong direct product theorems cannot be expected, and in fact even direct sum theorems (in which we do not care about the success probability but only about the scaling of the complexity with $k$) do not hold.
His main argument is that in the distributional complexity setting one can construct functions $f$ for which there is a ``hard core" of some size $\epsilon$ that cannot be ignored when allowing only error probability $\epsilon/3$ (making computing one instance hard), yet given $k$ instances
only roughly $\epsilon k$ of them will be in the hard core, and we can re-allocate most of our resources to those. By construction $f$ is trivial outside the hard core and we can easily solve the other instances. Altogether this approach uses roughly $\epsilon k$ times the resources needed for one instance while having very small overall error. The main conclusion of this example should be that when proving direct sum type statements in the distributional setting one should expect to lose a factor of $\epsilon$ in the complexity bound.

An incomplete list of examples of ``positive" results about DPT's are Nisan et al.'s~\cite{nrs:products}
strong direct product theorem for ``decision forests'', Parnafes
et al.'s~\cite{prw:productgcd} direct product theorem for
``forests'' of communication protocols, Shaltiel's strong direct
product theorems for ``fair'' decision trees and for the discrepancy
bound for communication complexity under the uniform distribution  \cite{shaltiel:sdpt}, Lee et al.'s analogous result for arbitrary
distributions \cite{lee:dptdisc}, Viola and Wigderson's extension to the multiparty case \cite{viola:xor}, Ambainis et al.'s SDPT for the quantum query complexity of symmetric functions \cite{ambainis:qdpt}, Jain et al.'s SDPT for subdistribution bounds in communication complexity \cite{jain:subdis}, Ben-Aroya et al.'s SDPT for the quantum one-way communication complexity of the Index function \cite{ben-aroya:hyper}, Impagliazzo et al.'s DPT for uniform circuits \cite{impagliazzo:dpt} and several more. In a similar vein are ``XOR"-lemmas like Yao's \cite{yao:xor}. ``Direct Sum" results which state that $k$ times the resources are needed without the success probability deterioration are also important in communication complexity, see \cite{kushilevitz&nisan:cc,barak:dsum}.

In this paper we focus on the Disjointness problem in communication complexity.
Suppose Alice has an $n$-bit input $x$ and Bob has an $n$-bit input $y$.
These $x$ and $y$ represent sets, and $\DISJ_n(x,y)=1$ iff those sets are disjoint.
Note that $\DISJ_n$ is the negation of $\NDISJ_n=\OR_n(x\wedge y)$, where
$x\wedge y$ is the $n$-bit string obtained by bitwise AND-ing $x$ and $y$.
In many ways, $\NDISJ_n$ plays a central role in communication
complexity. In particular, it is ``NP complete''~\cite{bfs:classes} in the communication complexity world.
The communication complexity of $\NDISJ_n$ has been well studied:
e.g.~it takes $\Th n$ bits of communication classically~\cite{ks:disj,razborov:disj}
and $\Th{\sqrt{n}}$ quantumly \cite{aaronson&ambainis:search,razborov:qdisj}.

For the case where Alice and Bob want to compute $k$ instances
of Disjointness, we establish a strong direct product theorem in Section~\ref{sec:dpt}:
\begin{quote}
{\bf SDPT for randomized communication complexity:}\\
Every randomized protocol that computes $\vectorres {\NDISJ_n} k$ using
$T\leq \beta k n$ bits of communication
has worst-case success probability $\sigma=2^{-\Om{k}}$.
\end{quote}
Note that the same result holds for $\DISJ_n$ by symmetry.
Previously, Klauck et al.~\cite{klauck:qsdpt} proved that the same success probability bound holds
when the communication is $\beta k\sqrt n$ (but even in the quantum case). The same bound was
obtained by Beame et al.~\cite{beame:sdpt} for randomized communication, and they give an SDPT for the {\em rectangle bound} under product distributions (under such distributions $\DISJ_n$ has complexity $\sqrt n$). The rectangle bound appears in the literature also under the name {\em corruption bound} \cite{yao:prob, klauck:thresh, beame:sdpt}. Klauck \cite{klauck:qcst} also showed
a weak DPT for the rectangle bound under {\it all} distributions, which implies that with communication $\beta n$
the success probability goes down exponentially in $k$.

Our approach is as follows. First we massage the problem in a very similar manner as in \cite{klauck:qsdpt}. This leads
to the problem of finding $k$ elements in the intersection of two $N$ bit strings. Since these can
easily be verified, we can assume that the protocol either gives up or produces correct outputs. We are interested in the tradeoff between success probability and communication.

The next step is to formulate a linear program that corresponds to a relaxation of an integer program expressing a convex combination of partitions of the communication matrix with the desired acceptance probabilities. Similar programs have been considered
before by Lov\'{a}sz \cite{lovasz:cc} and by Karchmer et al.~\cite{karchmer:fractional}, but have rarely been used to bound randomized communication complexity. The program expresses that we can
detect inputs $x,y$ with intersection size $k$ with ``high" probability, while not accepting inputs with smaller intersection size at all,
and, trivially but importantly, accepting the remaining inputs with probability at most 1. This extra constraint expresses the fact that we do not talk about covers of the communication matrix, but partitions. Unsurprisingly we prove the lower bound by exhibiting a solution to the dual. This approach is intimately related to the {\em smooth rectangle bound} explored in \cite{jain:partition}, see Section 1.3.

To prove feasibility of the dual solution we provide what we call the intersection sampling lemma.
This lemma is a generalization of Razborov's main lemma from \cite{razborov:disj} and follows  from it by a rather simple induction argument. The intersection sampling lemma states that (for suitable distributions) any rectangle that is large among the disjoint $x,y$ is also large for inputs that have intersection size $k$. This is true for every $k$, losing a $2^k$ factor. Razborov's Lemma is essentially the same statement for $k=1$.

\subsection{Applications}

\subsubsection{Communication-Space Tradeoffs}
Our main result has some applications to other problems.
First, we consider communication-space tradeoffs.
Research on communication-space tradeoffs
has been initiated by Lam et al.~\cite{lam:commtrade} in a
restricted setting, and by Beame et al.~\cite{beame:commtrade} in
a general model of space-bounded communication complexity. In the
setting of communication-space tradeoffs, players Alice and Bob are
space bounded circuits, and we are interested in the
communication cost when given particular space bounds.

We study the problems of \emph{Boolean matrix-vector product} and \emph{Boolean
matrix product}. In the first problem there are an $N\times N$
matrix $A$ (input to Alice) and a vector $b$ of dimension $N$ (input to Bob), and the goal is to
compute the vector $c=Ab$, where $c_i=\vee_{j=1}^n \left(A[i,j]\wedge
b_j\right)$. In the problem of matrix
multiplication two input matrices have to be multiplied with the analogous
Boolean product.

{\em Time}-space tradeoffs for Boolean matrix-vector multiplication have
been analyzed in an average case scenario by Abrahamson
\cite{abrahamson:booleantrade}, whose results give a worst case
lower bound of $TS=\Om{N^{3/2}}$ for classical algorithms. He
conjectured that a worst case lower bound of $TS=\Om{N^2}$
holds, which was later confirmed in \cite{klauck:qsdpt}.

Beame et al.~\cite{beame:commtrade} gave tight lower bounds for communi\-cation-space tradeoffs for the matrix-vector product and
matrix product over finite fields, but stated the complexity of Boolean
matrix-vector multiplication as an open problem. Klauck \cite{klauck:qcst} generalized the results for finite fields
to the quantum case, but also
showed the following lower bounds for the Boolean product and randomized protocols: for matrix-vector product $CS^2=\Om{N^2}$, and for matrix-matrix product $CS^2=\Om{N^3}$. Using our direct
product result we are now able to
show that any randomized protocol for matrix-vector product satisfies $CS=\Om{N^2}$, and for matrix-matrix product $CS=\Om{N^3}$.
These bounds match the trivial upper bounds.

\subsubsection{Multiparty Communication}

Consider the Nondisjointness problem in the 3 player num\-ber-on-the-forehead
setting, i.e., Alice sees inputs $y,z$ Bob sees $x,z$ and Charlie sees $x,y$. They have to decide whether there is an index $i$ such that $x_i=y_i=z_i=1$. Lee and Shraibman \cite{lee:disj} as well as Chattopadhyay and Ada \cite{chatto:multiparty} show that the randomized complexity of this problem is $\Omega(n^{1/4})$. Prior to these results larger bounds were shown for models in which the interaction between the players is restricted. In particular, in the model with one-way communication, Viola and Wigderson show a $\Omega(\sqrt n)$ lower bound \cite{viola:pointer}, and in the model, where Charlie sends a single message, followed by an arbitrary protocol between Alice and Bob, Beame et al.~\cite{beame:sdpt} show an $\Omega(n^{1/3})$ lower bound, which was later simplified by
 Ben-Aroya et al.~\cite{ben-aroya:hyper}. Using our main theorem we can show that the latter type of protocol actually needs communication $\Omega(\sqrt n)$.

\subsection{The Smooth Rectangle Bound}

Our main result is proved by giving a solution to the dual of a linear program. While this program is tailor made for the problem at hand, this is a general approach described e.g.~in \cite{lovasz:cc, karchmer:fractional}.

In \cite{lovasz:cc} Lov\'{a}sz in describes such a LP-based lower bound method for randomized protocols, which can be seen to be equivalent to the rectangle bound (for a proof see \cite{jain:partition}).
Adding a seemingly trivial constraint to the LP described by Lov\'{a}sz in gives a more powerful lower bound method (via the dual), by using the fact that protocols {\em partition} the inputs into rectangles instead of covering them. The lower bound method is similar to the rectangle bound, but allows the use of negative weights for a small fraction of the 1-inputs.
We refer to this enhanced LP-based lower bound method as the {\em smooth rectangle bound}, because it can be seen as a maximum of the rectangle bounds achievable by functions that are close to the function $f$ we are interested in. The smooth rectangle bound is defined and explored in \cite{jain:partition}. The linear program that we use to establish our main result also uses the partition property crucially, and in fact, Lemma 1 could not be established using the rectangle bound.

 The smooth rectangle bound relates to the rectangle bound similar to the way the generalized discrepancy method (introduced by \cite{sherstov:pattern,klauck:lbqcc} and named so in \cite{chatto:multiparty}) relates to the standard discrepancy bound (both methods are lower bounds on quantum communication, and the generalized discrepancy is in fact equivalent to Linial and Shraibman's (approximate) $\gamma_2$-measure \cite{linial:norms}).

We can pinpoint the power of the smooth rectangle bound (for Boolean functions) more closely by observing that it actually lower bounds {\em unambiguous AM-protocols}. Indeed our main result also holds for AM-protocols with ambiguity $2^{\epsilon k}$ as we state in Theorem 11.  Note that $\NDISJ_n$ has very efficient nondeterministic protocols (and so does its $k$-fold), so the lower bound really comes from the partition constraints. In particular we also show that any unambiguous AM-protocol for $\NDISJ_n$ needs linear communication (while nondeterministic protocols need communication $O(\log n))$. Note that proving lower bounds for unrestricted AM-protocols is an open problem.

\section{Preliminaries}

In this section we give some definitions of some of the models of communication we study. We refer to \cite{kushilevitz&nisan:cc} for more background in communication complexity.

\subsection{Some Definitions on Communication Complexity}
The protocols we consider are in the standard two-player model \cite{yao:cc, kushilevitz&nisan:cc} unless stated otherwise.
The randomized protocols we consider are public coin protocols. Success probability of a protocol is defined to be the probability over the coins to compute the correct output for a worst case input. Note that we require both players to agree on a common output.

A nondeterministic protocol for a Boolean function $f$ is a cover of the 1-inputs in the communication
matrix of $f$ with 1-chromatic rectangles, its cost is the logarithm of the number of rectangles used. Alternatively, a nondeterministic protocol can be viewed as a proof system, in which a prover sends a proof to Alice, after which Alice and Bob verify the proof. In a valid protocol for all 1-inputs there exists a proof that is accepted, and for all 0-inputs all proofs are rejected. The cost is the amount of communication between Alice and Bob. A nondeterministic protocol with ambiguity $t$ is a nondeterministic protocol in which each 1-input has no more than $t$ different proofs. For $t=1$ such protocols are called unambiguous.

Karchmer et al.~\cite{karchmer:witness} have shown that nondeterministic protocols with ambiguity $t$ have complexity at least $\Omega(\sqrt{D(f)}/t)$. Also the rank lower bound holds for unambiguous protocols.

In a computation of a $k$-tuple of Boolean functions by a nondeterministic protocol, the prover wants to convince Alice and Bob of the fact that $f(x_i,y_i)=1$ for as many $i$ as possible. Such a protocol is correct, if for all $x_1,y_1,\ldots, x_k,y_k$ such that $f(x_i,x_i)=1$ for all $i\in I\subseteq\{1,\ldots,k\}$ there is a proof such that Alice and Bob agree on output $o_1,\ldots, o_k$ with $o_i=1\iff i\in I$, while for no $i\not\in I$ there exists a proof such that $o_i=1$ will be an output. Note that in this definition we never require the prover to convince Alice and Bob of the fact that $f(x_i,y_i)=0$ for any position $i$, so this is genuine one-sided nondeterminism for many-output problems.

In other words every nondeterministic protocol with ambiguity $t$ is a collection of at most $2^c$ rectangles each labeled by an output sequence such that for each input $x_1,y_1,\ldots, x_k,y_k$ and each rectangle $R$ with output $o_1,\ldots,o_k$ containing that input: $o_i\leq f(x_i,y_i)$ for all $i$, and there exists a rectangle containing the input where $o_i=f(x_i,y_i)$ for all $i$. Furthermore each input is contained in at most $t$ such rectangles. The communication cost is then $c$.

An Arthur-Merlin communication protocol (first suggested in \cite{bfs:classes}) with ambiguity $t$ and communication $c$ is a convex combination of a set of nondeterministic protocols $P_i$, each occurring with probability $p_i$. Each nondeterministic protocol is a collection of at most $2^c$ rectangles each labeled by an output sequence and each input is contained in at most $t$ such rectangles per $P_i$.
We require that for each input $x_1,y_1,\ldots, x_k,y_k$ with probability at least $1-\epsilon$ the protocol $P_i$ has  $x_1,y_1,\ldots, x_k,y_k$ in some rectangle labeled $f(x_1,y_1),\ldots,f(x_k,y_k)$, whereas with probability at most $\epsilon$ a $P_i$ contains the input in a rectangle labeled with $o_j=1$ while $f(x_j,y_j)=0$ for some $i$. An AM-protocol with ambiguity 1 is called unambiguous (note that for different values of the public coin different proofs are allowed for the same input).

\subsection{Communicating Circuits}

In the standard model of communication complexity Alice and Bob are computationally unbounded
entities, but we are also interested in what happens if they have bounded memory, i.e.,
they work with a bounded amount of storage. To this end we model Alice and Bob as
communicating circuits. In short, these circuits place no restrictions on local gates, but require the number of bits stored locally to be bounded. Communication is the number of wires crossing between Alice and Bob's part of the circuit.

A pair of communicating circuits is actually a single
circuit partitioned into two parts. The allowed operations
are local computations and access to the inputs. Alice's part of the circuit may
read single bits from her input, and Bob's part of the circuit
may do so for his input. Otherwise arbitrary gates (of any fan-in) on the
locally available bits can be used.

The communication $C$ between the two
parties is simply the number of wires carrying bits that cross
between the two parts of the circuit. A pair of communicating
circuits uses space $S$, if the whole circuit works on $S$
bits storage.
In the problems we consider, the number of outputs is much larger
than the memory of the players. Therefore we use the following
output convention. The player who computes the value of an output
sends this value to the other player at a predetermined point in
the protocol, who is then allowed to "forget" the output.
Outputs have to be made in some specified order in the circuit, i.e., we expect the $i$th output
to be made at a specific gate.

\section{The Direct Product Theorem}\label{sec:dpt}

In this section we formally state and prove our main result.
\subsection{Massaging the Problem}

In this section we bring the $k$-fold $\NDISJ_n$ problem into anther form that will be easier to handle.
More precisely, we will consider the following three problems. We freely identify strings $x\in\{0,1\}^n$ with
the sets they are characteristic vectors of.

\begin{definition}
\begin{enumerate}
\item $\vectorres{\NDISJ_n}k$ is the problem, given $k$ pairs of strings $x_i,y_i$ of length $n$ each, to compute
 the $k$-tuple of function values of $\NDISJ_n$ on these.
\item $\vectorres{\Search_n}k$ is the problem, given $k$ pairs of strings $x_i,y_i$ of length $n$ each, to find
  indices $j_1,\ldots,j_k$, such that $x_i$ and $y_i$ intersect in $j_i$. If $x_i$ and $y_i$ are disjoint, output 0 for position $i$.
\item $\Search_{{N\choose k}}$ is the problem, given two strings $x,y$ of length $N$, to find $k$ indices $j_1,\ldots,j_k$, such that $x$ and $y$ intersect in all $j_i$. If $|x\cap y|<k$ output 0.
\end{enumerate}
\end{definition}

We will prove that problem 3) is hard in the following subsections and state the result now.

\begin{lemma} [Main]\label{lem:main}
There are constants $0<\alpha, \beta,\gamma\leq 1$ such that every randomized protocol with communication $\beta N$ for the problem $\Search_{{N\choose k}}$ with $k\leq \gamma N$ has success probability at most $2^{-\alpha k}.$
\end{lemma}

We now establish that the first two problems are also at least as hard as 3) by reductions very similar to the analogous reductions in \cite{klauck:qsdpt}.

\begin{theorem}[SDPT for Search]\label{th:ind-search}
There are constants $0<\alpha', \beta\leq 1$ such that every randomized protocol with communication $\beta kn$ for the problem $\vectorres{\Search_n}k$ has success probability at most $2^{-\alpha' k}$.
\end{theorem}

\begin{proof}
We show that a protocol for $\vectorres{\Search_n}{k}$ can be used to solve $\Search_{{N\choose K}}$ for $K=\alpha k/4$ and
$N=kn$. Fix a protocol $P$ for $\vectorres{\Search_n}{k}$ with success probability $\sigma$.
Now consider the following protocol that acts on $N$-bit inputs $x,y$:
\begin{enumerate}
\item Apply a uniformly random permutation $\pi$ to $x$ and to $y$.
\item Run $P$ on $\pi(x)$, $\pi(y)$.
\item If $P$ makes at least $\alpha k/4$ outputs $\neq 0$, then output any $\alpha k/4$ of them (after undoing $\pi$).
\item Otherwise output 0.
\end{enumerate}
This protocol $P'$ uses the same communication as $P$. Note that $P'$ will work correctly and solve  $\Search_{{N\choose \alpha k/4}}$, whenever $P$ makes no errors and at least $\alpha k/4$ positions $i$ with $x_i=y_i=1$ end up in different blocks after applying the permutation $\pi$ (so that they can be produced as outputs by $P$),  assuming that $\alpha k/4$ such positions exist.

The probability of at least $\alpha k/4$ positions $i$ with $x_i=y_i=1$ being in different blocks (assuming that so many exist) is at least
\[
{N \over N} \cdot {N-n \over N-1} \cdots {N-\alpha (k/4) n+1 \over N-\alpha (k/4)+1}
\ge \left( 1 - \alpha/4  \right)^{\alpha k/4}
.
\]

So the success probability of $P'$ is at least $\sigma\cdot(1-\alpha/4)^{\alpha k/4}$ which defines $\alpha'$ via
\[2^{-\alpha K}= 2^{-\alpha\cdot (\alpha k/4)}\geq 2^{-\alpha' k}\cdot(1-\alpha/4)^{\alpha k/4}.\]
This allows us to choose a constant $\alpha'>0$, since $(1-\alpha/4)^{\alpha/4}\geq 2^{-\alpha^2/8}$ for $0\leq\alpha<1$.

The above argument only works, if $\alpha k/4\leq \gamma N\iff n\geq \alpha/(4\gamma)$. Since the right hand side involves only constants the opposite case can be covered by assuming $n=O(1)$, i.e., we now have to show that solving many size $O(1)$ instances is hard. But when the communication is less than $\epsilon k=\Theta(kn)$, it can easily be shown via an information theoretic argument, that it is impossible to solve $\vectorres{\Search_n}k$ with better success than $2^{-\Omega(k)}$: under the uniform distribution the players don't communicate enough to agree on a set of $k$ outputs of sufficient entropy.
\end{proof}

\begin{theorem}[SDPT for $\NDISJ_n$]\label{th:vector-disj}
There are constants $0<\alpha'',\beta''\leq 1$ such that every randomized protocol
for $\vectorres{\NDISJ_n}k$ with $\beta'' kn$ communication has
success probability $\sigma \leq 2^{-\alpha'' k}$.
\end{theorem}

\begin{proof}
A protocol  $P$ for $\vectorres{\NDISJ_n}k$ with success probability $\sigma$
and communication $C\leq\beta'' kn$ can be used to build a protocol $P'$ for $\vectorres{\Search_n}k$
with slightly worse success probability:
\begin{enumerate}
\item Run $P$ on the original inputs and remember which blocks are accepted.
\item Run simultaneously (at most $k$) binary searches on the accepted blocks for a limited number of steps that halve the search space.
Stop after $s=2\log(1/\beta'')$ such steps. Each iteration is computed by running $P$ on
the parts of the blocks that are known to contain a position $j$ with $x_i(j)=y_i(j)=1$, halving the remaining
instance size each time.
\item Run the trivial protocol on each of the
remaining parts of the instances to look for an intersection there
(each remaining part has size $n/2^s$).
\end{enumerate}
This new protocol $P'$ uses communication
$(s+1)C+kn/2^s$ $=\Oh{\beta''\log(1/\beta'') kn}$. With probability at least $\sigma^{s+1}$, $P$ succeeds in all iterations,
in which case $P'$ solves $\vectorres{\Search_n}k$.

So setting $\beta''$ such that $\beta\geq\Oh{\beta''\log(1/\beta'')}$ and $\alpha''=\alpha'/(s+1)$ we get the desired reduction.
\end{proof}

\subsection{Our Approach}

In order to establish Lemma 1 we need to use a new approach to prove lower bounds in communication complexity. To see this we discuss the main techniques available. The rectangle bound cannot be used to prove Lemma 1, because there exist large monochromatic rectangles for all possible outputs except rejection, since the problem has low nondeterministic complexity: simply guess $k$ outputs and check. While large rectangles for the inputs that must be rejected need to have at least constant error due to Razborov's lower bound for the distributional complexity of $\DISJ_n$ \cite{razborov:disj}, simply considering those inputs cannot establish the error bound we seek. So the rectangle bounds appears to be unsuitable for our purpose. Note, however, that it might be possible to establish the strong direct product theorem for $\DISJ_n$ using the rectangle bound, but our approach via Lemma 1 cannot go this way.

Another main method for proving lower bounds for randomized communication complexity is the $\gamma_2$/generalized discrepancy method due to Linial and Shraibman/Sherstov \cite{linial:norms,sherstov:pattern} (and inspired by earlier work in \cite{klauck:lbqcc}). However, since these lower bounds work in the quantum setting, they cannot improve upon the tight quantum bounds in \cite{klauck:qsdpt}.

Finally one could employ information theoretic techniques like in \cite{bar-yossef:disj}. However, proving direct product statements with such methods seems to be difficult.

We show our main result using a technique based on linear programming. In the dual picture we still have to argue that all rectangles have certain properties, however, this time we allow positive and negative weights instead of a single hard distribution. This expresses the extra constraints given to us by the fact that a protocol partitions the communication matrix into rectangles instead of just covering it.

In the next subsections we first describe our linear program, then define a costly solution to the dual, and finally prove that this solution is feasible.

\subsection{The Linear Program}

In this section we provide a linear program, whose value gives a lower bound on the communication complexity of
solving $\Search_{{n\choose k}}$ with success probability $\sigma$. This will be our tool to establish Lemma \ref{lem:main}.

So consider any protocol for $\Search_{{n\choose k}}$ with success probability $\sigma$. We can assume that the protocol either rejects, or outputs $i_1,\ldots, i_k$. In the latter case we require that the inputs $x,y$ do actually intersect on those positions, or the other way around, that wrong outputs of this form have probability $0$. This we can assume, because Alice and Bob can simply check an output, before making it ``official". The communication overhead is just two bits to agree on the output being correct. Furthermore in this case every message sequence has a fixed particular set of outputs that Alice and Bob agree on, i.e., for any rectangle $R$ that corresponds to a leaf of the communication tree (for any value of the random coins) there are $k$ different positions $i_1,\ldots, i_k$ such that all inputs $x,y\in R$ intersect on them, or the protocol rejects (they need not be unique though). Otherwise the protocol would declare a non-rejecting output that is not correct for some inputs (note that by definition Alice and Bob must agree on an output in each terminal rectangle in the communication tree).

We can change such a protocol to a protocol with binary output in which inputs with intersection size $k$ are accepted with probability $\geq\sigma$, whereas all inputs with intersection size smaller than $k$ are accepted with probability 0.
Furthermore on all inputs acceptance happens with probability at most $1$. This latter trivial constraint is important in our proof.
The linear program is now as follows. We have real variables $w_R$ for {\it all} rectangles $R\subseteq\{0,1\}^n\times\{0,1\}^n$.

\begin{eqnarray}
&&\min \sum_R w_R  \hspace{2cm}\mbox{ s.t.} \\
&&\mbox{for all } x,y \mbox{ with }|x\cap y|<k: \sum_{R:x,y\in R}w_R=0\\
&&\mbox{for all } x,y \mbox{ with }|x\cap y|=k: \sum_{R:x,y\in R}w_R\geq\sigma\\
&&\mbox{for all } x,y \mbox{ with }|x\cap y|\geq k: \sum_{R:x,y\in R}w_R\leq 1\\
&&w_R\geq 0
\end{eqnarray}

Let $P$ be a randomized protocol with communication $c$ and success probability $\sigma$ for the problem of accepting inputs $x,y$ with $|x\cap y|=k$ while rejecting inputs $x,y$ with $|x\cap y|<k$ with certainty (acceptance probability on the other inputs does not matter).
$P$ can be used to create a solution to the above program with cost $2^c$: $P$ is a convex combination of deterministic protocols
$P_1,\ldots,P_m$ with probabilities $p_1,\ldots, p_m$, and each deterministic protocol $P_i$ corresponds to a partition of the inputs into $2^c$ rectangles. We restrict our attention to the rectangles on which protocols $P_i$ accept. The weight $w_R$ of a rectangle $R$ is the sum of the $p_i$ over all $P_i$ in which $R$ occurs as an accepting rectangle. Then for all inputs $x,y$ the value $\sum_{R:x,y\in R}w_R$ is
simply the acceptance probability of the protocol $P$, and the solution is feasible with cost $2^c$. Hence for any $\sigma$ the logarithm of the optimal cost of the program yields a lower bound on the necessary communication.

Recall that above we have not only required that the protocol $P$ accepts
inputs $x,y$ that intersect in exactly $k$ positions with some probability $\geq\sigma$, but we have also that for each accepting message sequence (i.e., each accepting rectangle $R$) there is a set of positions $I\subseteq\{1,\ldots,n\}$, $|I|=k$ such that for all inputs $x,y\in R$ we have $I\subseteq x\cap y$. Denote by ${\cal R}_v$ the set of all rectangles $R$ for which there is an $I\subseteq\{1,\ldots,n\}$, $|I|=k$ such
that all $x,y\in R$ satisfy  $I\subseteq x\cap y$. We can hence restrict the rectangles $R$ to come from  ${\cal R}_v$ in our LP. This also makes the constraints (2) superfluous.

We now take the dual of the program (with restricted rectangle set ${\cal R}_v$) and then show a lower bound by exhibiting a feasible solution of high cost.

\newpage
The dual is

\begin{eqnarray}
&&\max \sum_{x,y} \sigma\phi_{x,y}+\psi_{x,y} \hspace{2cm}\mbox{ s.t.} \\
&&\phi_{x,y}\geq 0\\
&&\psi_{x,y}\leq0\\
&&\mbox{ if } |x\cap y|\neq k\mbox{ then }\phi_{x,y}=0\\
&&\mbox{ for all }R\in {\cal R}_v:\sum_{x,y\in R} \phi_{x,y}+\psi_{x,y}\leq 1
\end{eqnarray}

The program asks us to put weights on the inputs, where inputs $x,y$ with intersection size $k$ should receive positive weights, and
some other inputs negative weights. The constraints demand that all rectangles in ${\cal R}_v$ either have small weight or contain enough negative weight to cancel most of the positive weight (we will discuss
this approach further in Section 5). However, we can only afford an overall amount of negative weight which is much smaller than the overall positive weight (by a factor of $\sigma$), because otherwise the objective function becomes negative.
Negative weights make it easier to satisfy the rectangle constraints (10), but deteriorate the cost function. Note that later on we will prove that all rectangles are either small, or the inequality in (10) can even be made negative. This is similar to the standard argument occurring with the usual rectangle bound: rectangles are either small, or they have large error. In the LP formulation both possibilities are rolled into one statement.

 Intuitively the LP formulation states that it is hard to cover the inputs with intersection size $k$ while keeping the partition constraints (4) satisfied. Note that the primal without (4), but keeping (2) has a very simple solution of cost $\exp(k\log n)$, even for $\sigma=1$. The issue with that solution is that it corresponds to a nondeterministic protocol, but not to a randomized one.

\subsection{The Solution}

Having found a dual program which will allow us to prove a lower bound, we start by defining distributions on inputs with different
intersection sizes in a similar way to \cite{razborov:disj}.

\begin{definition}

For $I=\{i_1,\ldots,i_k\}\subseteq\{1,\ldots,n\}$ with $|I|=k$ denote by $S_{I,n}$ the set of inputs $x,y\in\{0,1\}^n\times\{0,1\}^n$ such that $x\cap y=\{i_1,\ldots, i_k\}$. Furthermore let $T_{k,n}=\cup_{I:|I|=k} S_{I,n}$ denote the set of all inputs with intersection size $k$.

$\mu_{k,n,m}$ is a distribution on $\{0,1\}^n\times\{0,1\}^n$. All $x,y\not\in T_{k,n}$ have probability 0. Inputs in $T_{k,n}$  that also satisfy $|x|=|y|=m$ are chosen uniformly, i.e.,  with probability
\[\frac{1}{{n\choose m}{m\choose k}{n-m\choose m-k}}.\]
\end{definition}

An easy calculation shows
\begin{lemma}
\[\mu_{2k, n+k,m+k}(x,y)=\frac{{n\choose k}}{{n+k\choose 2k}}\cdot\mu_{k,n,m}(x',y'),\]
\[\mu_{k, n+k,m+k}(x,y)=\frac{1}{{n+k\choose k}}\cdot\mu_{0,n,m}(x',y'),\]
\[\mu_{k, n,m}(x,y)=\frac{1}{{n\choose k}}\cdot\mu_{0,n-k,m-k}(x',y'),\]
\[\mu_{k+1, n,m}(x,y)=\frac{n-k}{{n\choose k+1}}\cdot\mu_{1,n-k,m-k}(x',y'),\]
where $x'y'$ are inputs resulting from $x,y$, when $k$ intersecting positions are removed.
\end{lemma}

The solution to the dual program is based on the following intuition. Since the problem is symmetric, we should assign weights
uniformly for all inputs $x,y$ with a given intersection size (and set size). Naturally we put a good amount of positive weight on the inputs in $T_{k,n}$, and these are the only inputs with positive weights. We do not need to put negative weights on inputs with smaller intersection sizes,
since we already restricted the set of viable rectangles to ${\cal R}_v$, hence those inputs appear in no rectangle in the program. All we need to do is to find
a set of inputs to assign negative weights to, in order to enforce the rectangle constraints (the overall negative weight we can distribute is $\sigma$ times the overall positive weight, and we hope that the program stays feasible with a large objective function for $\sigma$ exponentially small in $k$). It turns out the $2k$-intersection inputs work just fine. This is because the $2k$-intersection inputs end up in many more rectangles than their weight suggests compared to the $k$-intersection inputs.

So we define a solution as follows (the input length is set to $n+k$ in the remainder of the section):

\begin{itemize}
\item The positive weight inputs are in $T_{k,n+k}$. Their weight is defined as $\phi_{x,y}=2^{\beta n}\mu_{k,n+k,m+k}(x,y)$.
\item The negative weight inputs are in $T_{2k,n+k}$. Their weight is $\psi_{x,y}=-2^{\beta n}2^{-\alpha k}\mu_{2k,n+k,m+k}(x,y)$.
\item For all other inputs $x,y:\phi_{x,y}=\psi_{x,y}=0$.
\end{itemize}
$\beta,\alpha>0$ are some constants that we choose later.
We can right away compute the value of this solution, before checking its feasibility. If we set $\sigma=2^{-\alpha k+1}$, then
the value is

\begin{eqnarray*}
&&\sum_{x,y} \sigma \phi_{x,y}+\psi_{x,y}\\&=&
\sum_{x,y\in T_{k,n+k}}\sigma2^{\beta n}\mu_{k,n+k,m+k}(x,y)\\&-&\sum_{x,y\in  T_{2k,n+k}}2^{\beta n}2^{-\alpha k}\mu_{2k,n+k,m+k}(x,y)\\&=&
2^{\beta n}2^{-\alpha k},
\end{eqnarray*}

 since both $\mu$'s are distributions. So for $\alpha k\leq(\beta/2) n$ we get a linear lower bound on the communication, and we will require $k\leq \gamma n/2$ for some $\gamma\leq\beta$ and set $\alpha=1/2$. Hence, all that remains to establish Lemma 1 is showing feasibility of our solution for these parameters.

The ``sign" constraints (7,8,9) are obviously satisfied, so the only thing we need to check are the rectangle constraints (10). The following lemma is the main ingredient of the proof.

\begin{lemma} [Intersection Sampling Lemma] \label{lem:isl}
There is a constant $\gamma>0$, such that
for each rectangle $R=A\times B\subseteq\{0,1\}^n\times\{0,1\}^n$ with $\mu_{0,n,m}(R)\geq2^{-\gamma n}$  and all $k\leq \gamma n/2$  we have $\mu_{k,n,m}(R)\geq\mu_{0,n,m}(R)/2^{k+1}$.
\end{lemma}

This lemma is a generalization of Razborov's main lemma in \cite{razborov:disj}, which is essentially the same statement for $k=1$.
We shall give the proof in the next section, however, now it's time to show that our solution to the dual program is feasible.

So let us check the rectangle constraints. If $R$ is a rectangle first suppose that $\mu_{k,n+k,m+k}(R)\leq 2^{-\beta n}$. In
this case $\sum_{x,y\in R} \phi_{x,y}\leq \sum_{x,y\in R\cap T_{k,n+k}} 2^{\beta n}\mu_{k,n+k,m+k}(x,y)\leq 1$.

Hence we need only worry about large rectangles $R\in{\cal R}_v$. For each such $R$ there is a set $I=\{i_1,\ldots, i_k\}$ of size $k$ such that all inputs $x,y$ in $R$ intersect on $I$. If we remove those positions from the universe $\{1,\ldots, n+k\}$ ($I$ is actually unique for all rectangles that contain inputs with positive weights at all) we can consider $R$ as a rectangle $R'$ in $\{0,1\}^{n}\times\{0,1\}^{n}$.
Clearly $\mu_{0,n,m}(R')\geq \mu_{k,n+k,m+k}(R)$, since all inputs in $R\cap T_{k,n+k}$ have a corresponding input in $R'\cap T_{0,n}$, and for
each $x,y$:$\mu_{0,n,m}(x,y)=\mu_{k,n+k,m+k}\cdot {n+k\choose k}$. So the intersection sampling lemma is applicable to $R'$  as long as we set $\beta=\gamma$ and $k\leq \gamma n/2$. 
The lemma tells us that $\mu_{k,n,m}(R')\geq\mu_{0,n,m}(R')/2^{k+1}$.

Consequently,
\begin{eqnarray}
\mu_{2k,n+k,m+k}(R)&=&\mu_{k,n,m}(R')\cdot\frac{{n\choose k}}{{n+k\choose 2k}}\\
&\geq&\mu_{0,n,m}(R')\cdot\frac{{n\choose k}}{{n+k\choose 2k}2^{k+1}}\\
&=&\mu_{k,n+k,m+k}(R)\cdot\frac{{n\choose k}{n+k\choose k}}{{n+k\choose 2k}2^{k+1}}\\
&\geq&\mu_{k,n+k,m+k}(R)\cdot\Omega(2^k/\sqrt k),
\end{eqnarray}
where (11) and (13) follow from Lemma 4, (12) from Lemma 5, and (14) using Sterling approximation.

So, surprisingly, the intersection sampling lemma lets us conclude that $R$ contains a lot more weight on $2k$-intersections inputs than on $k$-intersection inputs. Of course this is really a consequence of the fact that we forced the original protocol to be correct
in its (non--rejecting) outputs, and hence the fact that every rectangle we consider has one set of $k$ positions that all its inputs intersect in.

So \begin{eqnarray*}
\sum_{x,y\in R}\phi_{x,y}+\psi_{x,y}
&=&\sum_{x,y\in R\cap T_{k,n+k}}2^{\beta n}\mu_{k,n+k,m+k}(x,y)\\ &-&\sum_{x,y\in R\cap T_{2k,n}}2^{\beta n}\mu_{2k,n+k,m+k}(x,y)2^{-\alpha k}\\&\leq&0.\end{eqnarray*}

The rectangle constraints are satisfied and our program is indeed feasible.
We have the parameters $\beta=\gamma$, and $\sigma=2^{-\alpha k+1}$, and $\alpha=1/2$, as well as $k\leq\gamma n/2$.
Overall our solution to the dual proves that no protocol with communication $\beta n$ can solve $\Search_{{n+k\choose k}}$ with success better than $\sigma$, as long as $k\leq \gamma n/2$. By adjusting constants this proves Lemma 1.

\subsection{The Intersection Sampling Lemma}

In this section we prove Lemma 5 which we have used to establish the feasibility of the solution to the linear program exhibited in the
previous section.

The base of the induction proof will be provided by Raz\-borov's main lemma from \cite{razborov:disj} restated as follows:

\begin{fact}
There is a constant $\delta>0$, such that for all $m\in\{n/4-\delta n,\ldots,n/4\}$ and for every rectangle $R\subseteq\{0,1\}^n\times\{0,1\}^n$ with $\mu_{0,n,m}(R)\geq 2^{-\delta n}$ we have \[\mu_{1,n,m}(R)\geq\mu_{0,n,m}(R)/(3/2).\]
\end{fact}

The factor $3/2$ corresponds to error $2/5$ in the original statement, but it can be seen easily, that any error $1/2-\epsilon$ can be achieved in Razborov's proof by reducing the size of the rectangles considered suitably (i.e., by lowering the communication bound $\delta n$ considered). Also Razborov fixes $m=n/4$, but slightly smaller sets can be accommodated in the proof.\footnote{The proof needs to be adjusted in several ways. First of all, instead of mixing the distributions on intersection size 1 and 0 in the proportions 1/4 and 3/4 we need to mix them uniformly. Secondly, the constant 1/3 in the definition of $x$-bad can be replaced with a constant close to 1, and consequently the numbers in Claims 3 and 4 need to be adjusted. A bit more troublesome is allowing $m$ to be slightly smaller than $n/4$, since this makes Fact 2 false, although it remains approximately true, tilting all other estimates by $1+\delta$ factors.}

We prove the following statement by induction.

\begin{lemma} \label{islemma}
There is a constant $\gamma>0$, such that
for $m=n/4$ and every rectangle $R=A\times B\subseteq\{0,1\}^n\times\{0,1\}^n$ with $\mu_{0,n,m}(R)\geq2^{-\gamma n}$ and all $k\leq \gamma n/2$  we have
$$\mu_{k,n,m}(R)\geq\mu_{0,n,m}(R)/2^k-k\cdot 2^{-\delta (n-k+1)}.$$
\end{lemma}

In fact we choose $\gamma=\delta/3$ (and assume $k\leq \gamma n/2$). Then the above statement implies Lemma 5 as stated in the previous subsection.

\begin{proof}[of Lemma~\ref{islemma}]

Clearly the base of the induction over $k$ is true by Fact 6. So consider any rectangle  $R$, such that $\mu_{0,n,m}(R)\geq2^{-\gamma n}$ and assume
 $k\leq \gamma n/2$.

For all $I\subseteq\{1,\ldots,n\}$ with $|I|=k$ let's denote by $R_I$ the rectangle that is the intersection of $R$ with the rectangle that fixes
$x_i=y_i=1$ for all $i\in I$.
Now $R\cap\{T_{k,n}\cup\cdots\cup T_{n,n}\}=\cup_{I:|I|=k} R_I$. Furthermore every input $x,y\in T_{k+1,n}\cap R$ lies in exactly $k+1$ rectangles
\ $R_I$, while all inputs $x,y\in T_{k,n}\cap R$ lie in exactly one $R_I$. Hence

\[\mu_{k+1,n,m}(R)=\sum_{I:|I|=k}\mu_{k+1,n,m}(R_I)/(k+1).\]

Again we can reinterpret the $R_I$ as rectangles $R_I'$ in the set $\{0,1\}^{n-k}\times\{0,1\}^{n-k}$, because each $R_I$ has all its inputs
intersecting on the set $I$, so we only consider what happens on the other positions.

Note that $\mu_{0,n-k,m-k}(R_I')=\mu_{k,n,m}(R_I)\cdot {n\choose k}$ by Lemma 4, so we can conclude that
$\mu_{0,n-k,m-k}(R_I')$ is large whenever $\mu_{k,n,m}(R_I)$ is.\clearpage

Let ${\cal I}=\{I\subseteq\{1,\ldots,n\}:|I|=k\wedge\mu_{0,n-k,m-k}(R_I')\leq 2^{-\delta (n-k)}\}$. Then

\begin{eqnarray}
\sum_{I\in\cal I}\mu_{k,n,m}(R_I)\leq\sum_{I\in\cal I}\mu_{0,n-k,m-k}(R_I')/{n\choose k}\leq 2^{-\delta (n-k)}.\end{eqnarray}

Now

\begin{eqnarray}
\mu_{k+1,n,m}(R)&=&\sum_{I:|I|=k}\mu_{k+1,n,m}(R_I)/(k+1)\\
&=&\sum_{I:|I|=k}\mu_{1,n-k,m-k}(R_I')\cdot\frac{n-k}{{n\choose k+1}\cdot(k+1)}\\
&\geq&\sum_{I:|I|=k\wedge I\not\in\cal I}\mu_{1,n-k,m-k}(R_I')\cdot\frac{{n-k}}{{n\choose k+1}\cdot(k+1)}\\
&\geq&\sum_{I:|I|=k\wedge I\not\in\cal I}\mu_{0,n-k,m-k}(R_I')\cdot\frac{{(n-k)/(3/2)}}{{n\choose k+1}\cdot(k+1)}\\
&\geq&\sum_{I:|I|=k\wedge I\not\in\cal I}\mu_{k,n,m}(R_I)\cdot\frac{(n-k){n\choose k}}{{n\choose k+1}\cdot(k+1)\cdot 2}\\
&=&\sum_{I:|I|=k\wedge I\not\in\cal I}\mu_{k,n,m}(R_I)\cdot\frac{1}{2}\\
&\geq&\sum_{I:|I|=k}\mu_{k,n,m}(R_I)\cdot\frac{1}{2}-2^{-\delta (n-k)}\\
&\geq&\sum_{I:|I|=k}\mu_{0,n,m}(R_I)\cdot\frac{1}{2^{k+1}}-(k+1)2^{-\delta (n-k)}.
\end{eqnarray}

(17), (20) are via Lemma 4, (19) uses Fact 6, (22) is from (15), and (23) uses the induction hypothesis.
\end{proof}

\section{Applications}

\subsection{Communication-Space Tradeoffs for \\ Boolean Matrix Products}\label{seccommspace}

In this section we use the strong direct product result for
the communication complexity of Disjointness Theorem~\ref{th:vector-disj} to prove tight
communication-space tradeoffs.

\begin{theorem}\label{thcommspacemv}
Every bounded-error protocol with communication $C$ in which Alice and Bob have bounded space
$S$ and that computes the Boolean matrix-vector product, satisfies
$CS=\Om{N^{2}}$.
\end{theorem}

\begin{theorem}\label{thcommspacemm}
Every bounded-error protocol with communication $C$ in which Alice and Bob have bounded space
$S$ and that computes the Boolean matrix product, satisfies
$CS=\Om{N^{3}}$.
\end{theorem}

\begin{proof}[Proof of Theorem~\ref{thcommspacemv}.]
Alice receives a matrix $A$, and Bob a vector $b$ as
inputs. Given a circuit that multiplies these using communication
$C$ and space $S$ and that has success probability 1/2, we proceed to slice it. A slice of the circuit is a set of
consecutive gates in the circuit containing a limited amount of communication.
In communicating circuits the communication
corresponds to wires carrying bits that cross between Alice's
and Bob's part of the circuit. Hence we may cut the circuit after
$\beta N$ bits have been communicated and so on. Overall
there are $C/\beta N$ such circuit slices. Each starts with an
initial state computed by the previous part of the circuit. This
state on at most $S$ bits
may be replaced by the uniform distribution on $S$ bits. The effect is that the success probability
decreases to $(1/2)\cdot 1/2^S$, i.e., the outputs produced by the slice have at least this probability of being correct.

We want to employ the direct product theorem for
the communication complexity of $\NDISJ_{N/k}$ (for some $k$) to show that a protocol with the
given communication has success probability at most
exponentially small in the number of outputs it produces,
and so a slice can produce at most $\Oh{S}$
outputs. Combining these bounds with the fact that $N$
outputs have to be produced gives the tradeoff: $C/\beta N\cdot O(S)\geq N$.

To use the direct product theorem we restrict the inputs in the
following way: Suppose a protocol makes $k$ outputs. We partition
the vector $b$ into $k$ blocks of size $N/k$, and each block is
assigned to one of the $k$ rows of $A$ for which an output is
made. This row is made to contain zeroes outside of the positions
belonging to its block, and hence we arrive at a problem where
Nondisjointness has to be computed on $k$ instances of size $N/k$.
With communication $\beta N$, the success probability must
be exponentially small in $k$ due to Theorem~\ref{th:vector-disj}.
Hence $k=\Oh S$ is an upper bound on the number of outputs produced per slice.
\end{proof}

\begin{proof}[Proof of Theorem~\ref{thcommspacemm}.]
The proof uses the same slicing approach as in the other tradeoff
result. Note that we can assume that $S=\oo{N}$, since otherwise
the bound is trivial: the communication complexity without space restrictions is $\Omega(N^2)$.
Each slice contains communication $\beta
N$, and as before a direct product result showing that $k$
outputs can be computed only with success probability
exponentially small in $k$ leads to the conclusion that a slice
can only compute $\Oh S$ outputs. Therefore $(C/\beta N)\cdot \Oh S \ge
N^2$, and we are done.

Consider a protocol with $\beta N$ bits of communication making $k$ of the outputs. Each such output is the Boolean product of a row of $A$ and a column of $B$, and corresponds to a Nondisjointness problem. We
partition the universe $\{1,\ldots,N\}$ of the Nondisjointness
problems to be computed into $k$ mutually disjoint subsets
$U(i,j)$ of size $N/k$, each associated to an output $(i,j)$,
which in turn corresponds to a row/column pair $A[i]$, $B[j]$ in
the input matrices $A$ and $B$. Assume that there are $\ell$
outputs $(i,j_1),\ldots, (i,j_\ell)$ involving $A[i]$. Each output
is associated to a subset of the universe $U(i,j_t)$, and we set
$A[i]$ to zero on all positions that are not in one of these
subsets. Then we proceed analogously with the columns of $B$.

If the protocol computes on these restricted inputs, it
has to solve $k$ instances of Nondisjointness of size $n=N/k$ each,
since $A[i]$ and $B[j]$ contain a single block of size
$N/k$ in which both are not set to 0 if and only if $(i,j)$ is one
of the $k$ outputs. Hence Theorem~\ref{th:vector-disj} is applicable.
\end{proof}

\subsection{Multiparty}

\begin{theorem}
In the model where Charlie sends one message, followed by an arbitrary interaction between Alice and Bob, the 3-party Disjointness problem has randomized complexity $\Omega(\sqrt n)$.\end{theorem}

\begin{proof}
This proof idea is due to \cite{ben-aroya:hyper}.
Let $P$ be a protocol for the 3-party $\NDISJ_n$ problem with $\epsilon \sqrt n$ communication and error $1/3$.

We partition $\{1,\ldots,n\}$ into $\sqrt n$ blocks of size $\sqrt n$. Charlie's input $z$ is restricted to contain 1's in one
particular block, and 0's elsewhere. So in effect $z$ chooses one of $\sqrt n$ instances of $\NDISJ_{\sqrt n}$ given to Alice and Bob. Since Charlie's message does not depend on $z$, Alice and Bob may reuse it in $\sqrt n$ runs of $P$ in order to determine, for all $\sqrt n$ possible $z$, the value of all of their $\NDISJ_{\sqrt n}$ problems with overall communication $\epsilon n$.
For each block the error probability is $\leq\epsilon$. The expected number of blocks where the protocol fails is at most $2\epsilon \sqrt n$ with probability at least 1/2 by the Markov inequality. So for every input $x,y$ to Alice and Bob there is a message from Charlie which will make them give the wrong answer for at most $2\epsilon \sqrt n$ blocks with probability at least 1/2.

We may now simply replace Charlie's message by a uniformly random string which will deteriorate the probability of having at least $(1-2\epsilon)\sqrt n$ blocks correct to $2^{-\epsilon \sqrt n}\cdot 1/2$. We have found a 2 player protocol with communication $\epsilon n$ and the mentioned success probability to compute $(1-2\epsilon)\sqrt n$ instances of Nondisjointness correctly. In \cite{ben-aroya:hyper} a general argument is given that relates the success probability in this situation to the standard situation of an SDPT (in which success means all the outputs are correct). For small enough $\epsilon$ this contradicts our main result. The idea of such ``threshold DPT"'s is further
investigated in \cite{unger:dpt}.
\end{proof}

\section{The Linear Programming Bound and Limited Ambiguity}

A major tool to prove randomized communication complexity bounds is the rectangle bound (see \cite{klauck:thresh, beame:sdpt}).
The method was originally introduced by Yao \cite{yao:prob}, and its most prominent, but by no means only use is in Razborov's Disjointness bound \cite{razborov:disj}. Informally the rectangle bound states that all rectangles in the communication matrix are either small or have large error (under some hard distribution).

In \cite{lovasz:cc} Lov\'{a}sz describes the following LP to bound randomized communication complexity.

\begin{eqnarray}
&&\min \sum_R w_R  \hspace{7cm}\mbox{ s.t.} \\
&&\mbox{for } x,y \mbox{ with } f(x,y)=1: \sum_{R:x,y\in R}w_R\geq 1-\epsilon\\
&&\mbox{for } x,y \mbox{ with }f(x,y)=0: \sum_{R:x,y\in R}w_R\leq\epsilon\\\
&&w_R\geq 0
\end{eqnarray}

He takes the dual.

\begin{eqnarray}
&&\max \sum_{x,y:f(x,y)=1} (1-\epsilon)\phi_{x,y}
\quad\quad+\sum_{x,y:f(x,y)=0}\epsilon\phi_{x,y} \hspace{1.1 cm}\mbox{ s.t.} \\
&&\mbox{for } x,y \mbox{ with } f(x,y)=1: \phi_{x,y}\geq 0\\
&&\mbox{for } x,y \mbox{ with } f(x,y)=0: \phi_{x,y}\leq 0\\
&&\mbox{for all }R:\sum_{x,y\in R} \phi_{x,y}\leq 1
\end{eqnarray}

Note that here $R$ ranges over all rectangles in the communication matrix. One can now prove a lower bound by
exhibiting a solution to the dual. Let $\phi(x,y)$ describe such a solution.

The optimum of this program characterizes by the (one-sided) rectangle bound, as shown in \cite{jain:partition}.

Instead of proceeding like Lov\'{a}sz, who relaxes (31) (and obtains the spectral discrepancy bound, which can be exponentially smaller), we note the absence of the ``trivial" constraint
\begin{equation}
\mbox{for all } x,y \mbox{ with }f(x,y)=1: \sum_{R:x,y\in R}w_R\leq1. \end{equation}

The primal Lov\'{a}sz LP augmented with (32) obviously also gives a lower bound on randomized communication. This method is the smooth rectangle bound described in \cite{jain:partition}.
Consider the dual of the augmented program.

\begin{eqnarray}
&&\max \sum_{x,y:f(x,y)=1} (1-\epsilon)\phi_{x,y} +\sum_{x,y:f(x,y)=0}\epsilon\phi_{x,y} +\sum_{x,y:f(x,y)=1} \psi_{x,y}\hspace{2cm}\mbox{ s.t.} \\
&&\mbox{for all } x,y \mbox{ with } f(x,y)=1: \phi_{x,y}\geq 0;\psi_{x,y}\leq 0\\
&&\mbox{for all } x,y \mbox{ with } f(x,y)=0: \phi_{x,y}\leq 0,;\psi_{x,y}= 0\\
&&\mbox{ for all }R:\sum_{x,y\in R} \phi_{x,y}+\psi_{x,y}\leq 1
\end{eqnarray}

Is the smooth rectangle bound really stronger than the standard one?
Let us find out the strongest type of communication protocol that we can still lower bound. Consider unambiguous AM-protocols, i.e., convex combinations of partitions with bounded error, see Section 2. It is easy to see that the LP with constraint (32) lower bounds such protocols.
Note that $N(\NDISJ_n)=O(\log n)$, and hence also $AM(\NDISJ_n)=O(\log n)$. However, we can show that unambiguous $AM$ protocols for $\NDISJ_N$ need linear communication. To prove that the smooth rectangle bound is large for $\NDISJ_n$
we can restrict the set of rectangles to ${\cal R}_v=\{R:\exists i: x,y\in R\Rightarrow i\in x\cap y\}$. This is possible by the same argument as Theorem 3, using binary search for a limited number of iterations followed by a trivial protocol for the resulting small instance of Nondisjointness, to get a protocol for the problem of finding an $i$ with $x_i=y_i=1$. Again we can assume that we do not accept without having seen a witness $i$. We define a solution to the dual as follows:

\begin{itemize}
\item Inputs in $T_{0,n}$ have weight $\phi_{x,y}=-\infty$.
\item Inputs in $T_{1,n}$ have weight $\phi_{x,y}=2^{\beta n}\mu_{1,n,n/4}(x,y)$.
\item  Inputs in $T_{2,n}$ have weight $\psi_{x,y}=-2^{\beta n}\mu_{2,n,n/4}(x,y)\cdot3/4$.
\item All other inputs have weight 0.
\end{itemize}

The cost of this solution is $2^{\beta n}((1-\epsilon)-3/4)$. The sign constraints are satisfied. Now consider a rectangle $R\in {\cal R}_v$. Let $R'$ denote the rectangle in which we ignore its intersection position $i$. Then
\begin{eqnarray}
\mu_{2,n,m}(R)&=&\mu_{1,n-1,m-1}(R')\cdot (n-1)/{n\choose 2}\\
&\geq&\mu_{0,n-1,m-1}(R')\cdot (n-1)/(1.5{n\choose 2})\\
&=&\mu_{1,n,m}(R)\cdot  (n-1)n/(1.5{n\choose 2})\\
&=&\mu_{1,n,m}(R)\cdot 4/3.\end{eqnarray}

Above we use that $R$ is large (otherwise (36) is trivial)  and hence $\mu_{0,n-1}(R')\geq n\cdot \mu_{1,n}(R)\geq n 2^{-\beta n}\geq 2^{-\delta (n-1)}$ and employ Fact 6 in (38).

Then \[\sum_{x,y\in R:f(x,y)=1} \phi_{x,y}+\sum_{x,y\in R:f(x,y)=0} \phi_{x,y}\leq2^{\beta n}\mu_{1,n,m}(x,y)(1-3/4\cdot 4/3)=0.\]

This shows that any unambiguous AM-protocol for $\NDISJ_n$ needs communication $\Omega(n)$. It is known \cite{karchmer:witness} that nondeterministic protocols with ambiguity $t$ need communication $\sqrt{D(f)}/t$ for the deterministic complexity $D$, and the rank lower bound is also known to hold for unambiguous nondeterministic protocols. However, these methods do not allow errors, and the first bound cannot achieve linear lower bounds at all (the approximate rank cannot give better bounds than $\sqrt n$ either since it lower bounds quantum protocols \cite{razborov:qdisj}).

We can also generalize our main result Theorem 3 in a similar fashion:

\begin{theorem}
Assume an AM-protocol with ambiguity $2^{\epsilon k}$ computes $k$ instances of $NDISJ_n$. Then the success probability of the protocol (over the random bits) is at most $2^{\Omega(-k)}$ unless the communication is at least $\beta kn$.\end{theorem}

The proof of this statement is practically identical to the proof of Theorem 3. Going through the reductions in section 3 we see
that they need only a constant repetition of the original protocol, and a polynomial increase in the ambiguity.

When we come to the search problem and the linear programming formulation note that we have to replace the right
hand side of constraint (4) by $\leq 2^{\epsilon k}$. In the dual this increases the gap between $\phi$'s and $\psi$'s in the objective function, but that gap is already exponentially large in $k$, so nothing in the construction really changes.

Note that this bound is quite tight, since with ambiguity $2^{O(k)}$ we can reduce the communication to any fraction of $kn$, and with ambiguity $n^k$ the communication collapses to $O(k\log n)$ even without any error.

\subsection*{Acknowledgments}
  I thank Rahul Jain and Shengyu Zhang for insightful discussions. The idea of a ``smooth" rectangle bound originated in discussion between us. Ronald de Wolf gave very helpful comments to an earlier version of this paper.

\clearpage
\bibliographystyle{alpha}

\begin{thebibliography}{BNRW03}

\bibitem[AA03]{aaronson&ambainis:search}
S.~Aaronson and A.~Ambainis.
\newblock Quantum search of spatial regions.
\newblock In {\em Proceedings of 44th IEEE FOCS}, pages 200--209, 2003.
\newblock quant-ph/0303041.

\bibitem[Abr90]{abrahamson:booleantrade}
K.~Abrahamson.
\newblock A time-space tradeoff for {B}oolean matrix multiplication.
\newblock In {\em Proceedings of 31st IEEE FOCS}, pages 412--419, 1990.


\bibitem[ASW09]{ambainis:qdpt}
Andris Ambainis, Robert Spalek, Ronald de Wolf.
\newblock A New Quantum Lower Bound Method, with Applications to Direct Product Theorems and Time-Space Tradeoffs.
\newblock In{\em Algorithmica}, 55(3): 422--461, 2009.

\bibitem[BFS86]{bfs:classes}
L.~Babai, P.~Frankl, and J.~Simon.
\newblock Complexity classes in communication complexity theory.
\newblock In {\em Proceedings of 27th IEEE FOCS}, pages 337--347, 1986.

\bibitem[BKKS04]{bar-yossef:disj}
Ziv Bar-Yossef, T. S. Jayram, Ravi Kumar, D. Sivakumar.
\newblock An information statistics approach to data stream and communication complexity.
\newblock In {\em J. Comput. Syst. Sci.}, 68(4): 702--772, 2004.


\bibitem[BBCR10]{barak:dsum}
Boaz Barak, Mark Braverman, Xi Chen, and Anup Rao.
\newblock How to Compress Interactive Communication.
\newblock In {\em STOC 2010.}


\bibitem[BTY94]{beame:commtrade}
P.~Beame, M.~Tompa, and P.~Yan.
\newblock Communication-space tradeoffs for unrestricted protocols.
\newblock {\em SIAM Journal on Computing}, 23(3):652--661, 1994.


\bibitem[BPSW06]{beame:sdpt}
Paul Beame, Toniann Pitassi, Nathan Segerlind, Avi Wigderson.
\newblock A Strong Direct Product Theorem for Corruption and the Multiparty Communication Complexity of Disjointness.
\newblock In {\em Computational Complexity}, 15(4): 391--432, 2006.


\bibitem[BRW08]{ben-aroya:hyper}
Avraham Ben-Aroya, Oded Regev, Ronald de Wolf.
\newblock A Hypercontractive Inequality for Matrix-Valued Functions with Applications to Quantum Computing and LDCs.
\newblock In {\it FOCS 2008}, pages 477--486, 2008.


\bibitem[CA08]{chatto:multiparty}
A. Chattopadhyay, A. Ada.
\newblock Multiparty Communication Complexity of Disjointness.
\newblock ECCC Technical Report 15(002), 2008.


\bibitem[IJKW08]{impagliazzo:dpt}
R. Impagliazzo, R. Jaiswal, V. Kabanets, A. Wigderson.
\newblock  Uniform Direct Product Theorems: Simplified, Optimized, and Derandomized.
\newblock In: {\em STOC 2008}, pages 579--588, 2008.


\bibitem[JKN08]{jain:subdis}
Rahul Jain, Hartmut Klauck, Ashwin Nayak.
\newblock Direct product theorems for classical communication complexity via subdistribution bounds.
\newblock In: {\em STOC 2008}, pages 599--608, 2008.


\bibitem[JK10]{jain:partition}
Rahul Jain, Hartmut Klauck.
\newblock The Partition Bound for Classical Communication Complexity and Query Complexity.
\newblock To appear in {\em IEEE Conference on Computational Complexity 2010}. See arXiv:0910.4266.


\bibitem[JKS03]{jayram:app}
T. S. Jayram, Ravi Kumar, D. Sivakumar.
\newblock Two applications of information complexity.
\newblock In {\em  STOC 2003}, pages 673--682, 2003.


\bibitem[KS92]{ks:disj}
B.~Kalyanasundaram and G.~Schnitger.
\newblock The probabilistic communication complexity of set intersection.
\newblock {\em SIAM Journal on Discrete Mathematics}, 5(4):545--557, 1992.

\bibitem[KNSW94]{karchmer:witness}
Mauricio Karchmer, Ilan Newman, Michael E. Saks, Avi Wigderson.
\newblock Non-Deterministic Communication Complexity with Few Witnesses.
\newblock In {\em J. Comput. Syst. Sci.}, 49(2): 247--257, 1994.


\bibitem[KKN95]{karchmer:fractional}
Mauricio Karchmer, Eyal Kushilevitz, Noam Nisan.
\newblock Fractional Covers and Communication Complexity.
\newblock In {\em SIAM J. Discrete Math.}, 8(1): 76--92, 1995.

\bibitem[K03]{klauck:thresh}
Hartmut Klauck.
\newblock Rectangle Size Bounds and Threshold Covers in Communication Complexity.
\newblock In: {\em IEEE Conference on Computational Complexity 2003}, pages 118--134, 2003.

\bibitem[K04]{klauck:qcst}
Hartmut Klauck.
\newblock Quantum and Classical Communication-Space Tradeoffs from Rectangle Bounds.
\newblock In {\em FSTTCS 2004}, pages 384--395, 2004.

\bibitem[K07]{klauck:lbqcc}
Hartmut Klauck.
\newblock Lower Bounds for Quantum Communication Complexity.
\newblock {\em SIAM J. Comput.}, 37(1): 20--46, 2007.

\bibitem[KSW07]{klauck:qsdpt}
Hartmut Klauck, Robert Spalek, Ronald de Wolf.
\newblock Quantum and Classical Strong Direct Product Theorems and Optimal Time-Space Tradeoffs.
\newblock In {\em SIAM J. Comput.}, 36(5): 1472--1493, 2007.

\bibitem[KN97]{kushilevitz&nisan:cc}
E.~Kushilevitz and N.~Nisan.
\newblock {\em Communication Complexity}.
\newblock Cambridge University Press, 1997.


\bibitem[LTT92]{lam:commtrade}
T.W. Lam, P.~Tiwari, and M.~Tompa.
\newblock Trade-offs between communication and space.
\newblock {\em Journal of Computer and Systems Sciences}, 45(3):296--315, 1992.
\newblock Earlier version in STOC'89.

\bibitem[LSS08]{lee:dptdisc}
Troy Lee, Adi Shraibman, Robert Spalek.
\newblock A Direct Product Theorem for Discrepancy.
\newblock {\em IEEE Conference on Computational Complexity},  pages 71--80, 2008.


\bibitem[LS09]{lee:disj}
Troy Lee, Adi Shraibman.
\newblock Disjointness is Hard in the Multiparty Number-on-the-Forehead Model.
\newblock {\em Computational Complexity}, 18(2): 309--336, 2009.


\bibitem[LiS09]{linial:norms}
Nati Linial, Adi Shraibman.
\newblock Lower bounds in communication complexity based on factorization norms.
\newblock {\em Random Struct. Algorithms}, 34(3), pages 368--394, 2009.

\bibitem[L90]{lovasz:cc}
L. Lov\'{a}sz.
\newblock Communication Complexity: A Survey.
\newblock In {\em Paths, Flows, and VLSI Layout, edited by B. H. Korte}, Springer, 1990.

\bibitem[NRS94]{nrs:products}
N.~Nisan, S.~Rudich, and M.~Saks.
\newblock Products and help bits in decision trees.
\newblock In {\em Proceedings of 35th FOCS}, pages 318--329, 1994.

\bibitem[NS94]{nisan&szegedy:degree}
N.~Nisan and M.~Szegedy.
\newblock On the degree of {B}oolean functions as real polynomials.
\newblock {\em Computational Complexity}, 4(4):301--313, 1994.
\newblock Earlier version in STOC'92.


\bibitem[PRW97]{prw:productgcd}
I.~Parnafes, R.~Raz, and A.~Wigderson.
\newblock Direct product results and the {GCD} problem, in old and new
  communication models.
\newblock In {\em Proceedings of 29th ACM STOC}, pages 363--372, 1997.

\bibitem[Raz92]{razborov:disj}
A.~Razborov.
\newblock On the distributional complexity of disjointness.
\newblock {\em Theoretical Computer Science}, 106(2):385--390, 1992.

\bibitem[Raz03]{razborov:qdisj}
A.~Razborov.
\newblock Quantum communication complexity of symmetric predicates.
\newblock {\em Izvestiya of the Russian Academy of Science, mathematics},
  67(1):159--176, 2003.
\newblock quant-ph/0204025.


\bibitem[Sha01]{shaltiel:sdpt}
R.~Shaltiel.
\newblock Towards proving strong direct product theorems.
\newblock In {\em Proceedings of 16th IEEE Conference on Computational
  Complexity}, pages 107--119, 2001.


\bibitem[S08]{sherstov:pattern}
Alexander A. Sherstov.
\newblock The pattern matrix method for lower bounds on quantum communication.
\newblock {\em STOC 2008}, pages 85--94, 2008.

\bibitem[U09]{unger:dpt}
Falk Unger.
\newblock A Probabilistic Inequality with Applications to Threshold Direct-product Theorems.
\newblock {\em FOCS 2009}, 2009.

\bibitem[VW07]{viola:pointer}
E. Viola and A. Wigderson.
\newblock One-way multi-party communication lower bound for
pointer jumping with applications.
\newblock In: {\em Proceedings of the 48th IEEE Symposium on
Foundations of Computer Science}, 2007.

\bibitem[VW08]{viola:xor}
Emanuele Viola, Avi Wigderson.
\newblock Norms, XOR Lemmas, and Lower Bounds for Polynomials and Protocols.
\newblock In {\em Theory of Computing} 4(1): 137--168, 2008.

\bibitem[Yao79]{yao:cc}
A.~C-C. Yao.
\newblock Some Complexity Questions Related to Distributive Computing .
\newblock In {\em Proceedings of STOC 1979}, pages 209--213, 1979.

\bibitem[Yao82]{yao:xor}
A.~C-C. Yao.
\newblock Theory and applications of trapdoor functions.
\newblock In {\em Proceedings of 23rd IEEE FOCS}, pages 80--91, 1982.

\bibitem[Yao83]{yao:prob}A.~C-C.~Yao. \newblock Lower Bounds by Probabilistic
 Arguments. \newblock {\it 24th IEEE Symp.~Foundations of Computer Science}, pp.~420--428, 1983.

\end{thebibliography}

\newcommand{\etalchar}[1]{$^{#1}$}

\end{document}